\documentclass[12pt]{article}

\usepackage[margin=1.25in]{geometry}
\usepackage{mathrsfs, mathtools}
\usepackage{latexsym, amssymb, amsmath, amsthm}
\usepackage[round]{natbib}

\usepackage{bm}
\usepackage[title,toc,titletoc]{appendix}
\usepackage{titlesec}
\titleformat{\section}[block]{\large\normalfont\filcenter}{\thesection.}{.5em}{}
\titleformat{\subsection}[block]{\normalfont\bfseries\filright}{\thesubsection.}{.5em}{}

\usepackage{graphicx}
\usepackage{caption,subcaption}
\usepackage{comment}

\usepackage[T1]{fontenc}
\usepackage[utf8]{inputenc}
\usepackage{mathtools, slashed}

\usepackage[T1]{fontenc}
\usepackage[utf8]{inputenc}

\usepackage{setspace}
\usepackage{enumitem}
\usepackage[multiple, flushmargin]{footmisc}
\usepackage[colorlinks=true,linkcolor = blue, citecolor=blue, hyperfootnotes=false, hyperindex,breaklinks]{hyperref}
\usepackage[dvipsnames]{xcolor}
\usepackage{soul}

\usepackage{tikz}
\usetikzlibrary{positioning}

\usepackage{pgfplots}
\pgfplotsset{compat=1.14}
\usetikzlibrary{patterns, decorations.pathreplacing, arrows}

\tikzset{
  every node/.style    = {
    text centered,
    line width = .5,
    anchor = center,
  },
  every label/.style   = {
    fill = white, anchor = mid,
  },
  every path/.style   = {
    > = stealth
  },
  point/.style args   = {(#1)#2}{
    rounded corners,
    fill = white,
    minimum height = 20,
    minimum width = 10,
    label = { [name = #1] above:#2 },
  },
  point a/.style args   = {(#1)#2}{
    rounded corners,
    fill = white,
    minimum height = 10,
    minimum width = 20,
    label = { [name = #1] right:#2 },
  },
  point b/.style args   = {(#1)#2}{
    rounded corners,
    fill = white,
    minimum height = 10,
    minimum width = 75,
    label = { [name = #1] right:#2 },
  },
}

\usepackage{etoolbox}

\newtheorem{lemma}{Lemma}
\newtheorem{proposition}{Proposition}

\theoremstyle{definition}
\newtheorem{definition}{Definition}

\newtheorem{remark}{Remark}
\newtheorem{example}{Example}

\newcommand{\norm}[1]{\left\langle#1\right\rangle}
\newcommand{\abs}[1]{\left| #1 \right|}

\DeclareMathOperator*{\argmin}{arg\,min}
\DeclareMathOperator*{\argmax}{arg\,max}

\usepackage{cleveref}
\usepackage{quoting}
\usepackage{bibentry}

\usepackage{lmodern}

\begin{document}
\title{Diversity as Majorization\thanks{This paper subsumes parts of ``Diversity in Choice as Majorization,'' an extended abstract of which appears in the \emph{Proceedings of the 26th ACM Conference on Economics and Computation}. For helpful comments and discussions, we are grateful to Battal Dogan, Lars Ehlers, Matthew Elliott, Fuhito Kojima, Bobby Pakzad-Hurson, Alexander Westkamp, and seminar audiences at University of Bonn, Brown University, Durham Market-Design Workshop, EC'25, University of Montreal,  Recent Advances in School Choice 2025, RUD, SEA, SITE-Market Design, University of Tokyo, and ZEW Berlin. Mekonnen also thanks the Cowles Foundation for Research in Economics, where much of this research was conducted.}}
\author{Federico Echenique\thanks{Department of Economics, UC Berkeley. Contact: \href{mailto:fede@econ.berkeley.edu}{fede@econ.berkeley.edu}} \hspace*{1.25em} Teddy Mekonnen\thanks{Department of Economics, Brown University. Contact: \href{mailto:mekonnen@brown.edu}{mekonnen@brown.edu}}\hspace*{1.25em} M. Bumin Yenmez\thanks{Department of Economics, WUSTL and Business School, Durham University (UK). Contact: \href{mailto:bumin@wustl.edu }{bumin@wustl.edu}.}}

\maketitle
\setcounter{page}{0} \thispagestyle{empty}
\begin{abstract}
How should institutions compare group diversity, and which group should they select when they value diversity and merit? We take a target-based approach that evaluates the entire group composition without treating any type as intrinsically diversity-enhancing. Because different diversity indices may rank groups differently, we instead adapt majorization to construct an ordinal diversity preorder. We show that its maximally diverse selections are exactly those maximizing every index in a broad class. This characterization yields a reserve-and-quota policy that selects a maximally diverse group and, among such groups, the highest-merit agents. Any alternative is less diverse, less meritorious, or both.

\end{abstract}


\newpage
\begingroup
\allowdisplaybreaks
\section{Introduction}\label{sec:introduction}
Institutions often seek to form diverse groups. Firms seek to hire employees across fields of expertise, graduate programs wish to admit students across research areas, and public agencies seek to form committees whose members have varying socioeconomic or political backgrounds. But what exactly makes one group more diverse than another?

One notion, familiar from discussions of affirmative-action policies, equates greater representation of designated types with greater diversity. However, this notion does not specify when a type is under- or over-represented. We instead take a target-based approach that evaluates a group's composition relative to a benchmark---such as the population served by the institution---without treating any one type as intrinsically diversity-enhancing. 

Specifically, we propose a diversity preorder based on \textit{majorization} \citep{hardy1952inequalities}, but adapted to a primitive distributional target. We show that the maximally diverse groups under this preorder are exactly those that maximize every diversity index in a broad class, including target-adjusted versions of the Gini--Simpson index, Shannon entropy, R\'enyi entropy, and the Berger--Parker index. While different indices can differ in how they rank groups, our preorder identifies groups that are robustly optimal across different diversity indexes.

We apply our diversity preorder to study a second question: when an institution ranks agents according to merit, priority, or some other measure of quality, what selection policy yields a diverse group while favoring the highest-ranked agents? Here, we endogenously construct a reserve-and-quota policy that, together with the merit ranking, selects a group on the diversity-merit Pareto frontier.

The connection between diversity and majorization is motivated by the latter's characterization through \emph{Robin Hood transfers}. As an illustration, suppose our benchmark is the socioeconomic composition of a population, which is uniform across low $(\ell)$, middle $(m)$, and high $(h)$ income. A representative group of nine agents would therefore contain three agents of each type, yielding the target distribution $(3,3,3)$.  Consider a group $A$ consisting of one type-$\ell$, four type-$m$, and four type-$h$ agents, with distribution $(1,4,4)$. Relative to the target, type $m$ is over-represented while type $\ell$ is under-represented. A Robin Hood transfer reduces the imbalance between these two types by replacing one type-$m$ agent with one type-$\ell$ agent while leaving the number of type-$h$ agents unchanged. This produces a new group $B$ with type distribution $(2,3,4)$. Since $B$ is more representative of the broader population than $A$, we say that $B$ is more diverse than $A$. At the same time, because $B$ is obtained from $A$ through a Robin Hood transfer, its distribution is \emph{majorized} by that of $A$. More generally, one distribution is majorized by another if and only if it can be obtained from the latter through a sequence of Robin Hood transfers. Majorization provides a natural framework for comparative diversity when the benchmark is equal representation across types.

In practice, however, the relevant benchmark need not be uniform: We extend Robin Hood transfers to an arbitrary target composition $r$, yielding our notion of \emph{$r$-targeting diversity}. The target $r$ is a model primitive. It determines which types are over- or under-represented in a given group. For example, suppose $r$ is $1/9$ type-$\ell$, $4/9$ type-$m$, and $4/9$ type-$h$. Then, for a group of nine agents, the target distribution is $(1,4,4)$. Group $A$ therefore matches the target exactly, whereas $\ell$ is over-represented, and $m$ under-represented, in $B$. We say that $A$ is more $r$-diverse than $B$, reversing their ranking under the uniform benchmark.


Our first result formalizes the connection between $r$-targeting diversity and majorization. We define a simple transformation of distributions such that a group $A$ is more $r$-diverse than another group $B$ if and only if the transformed distribution of $A$ is majorized by that of $B$. This equivalence identifies the class of Schur-concave indices, which includes the indices listed earlier, that are monotone with respect to $r$-targeting diversity: whenever $A$ is more $r$-diverse than $B$, every index in this class assigns a weakly higher value to $A$. Because the $r$-targeting diversity preorder is incomplete, different indices may rank incomparable groups differently.

We then consider selecting $q$ agents when some types may be scarce and the target may be incompatible with integer selection. These issues matter especially when $q$ is small---for example, when selecting a handful of committee members for a public agency board---because replacing one agent can substantially change the group's composition. We characterize the \emph{$r$-diversity frontier}, the feasible selections that are maximally $r$-diverse.

The frontier displays significant structure, despite the incompleteness of the diversity preorder. Any two frontier groups are equally $r$-diverse, while every frontier group is strictly more $r$-diverse than every non-frontier group. Thus, the frontier forms a top equivalence class, and incomparability arises only among non-frontier groups. This conclusion does not follow from maximality alone, since an incomplete preorder may have incomparable maximal elements. Our main robustness result further shows that every strictly Schur-concave diversity index has exactly the $r$-diversity frontier as its set of feasible maximizers, making the choice of  index immaterial.



Finally, we address the second question: how should a decision maker select among the most diverse groups while favoring higher-merit agents? Using the $r$-diversity frontier, we endogenously construct type-specific reserves and quotas. The policy fills each type's reserved slots with its highest-merit agents and then fills the remaining slots by merit, subject to the quotas. We show that the resulting group is maximally $r$-diverse and that every alternative is either strictly less $r$-diverse or merit-dominated by the policy's outcome.



\subsection{Related Literature}
A large literature, beginning with \citet{abdulson03} and often motivated by controlled school choice, studies diversity in selection through type-specific reserves or quotas. \cite{hayeyi13} introduce minority reserves and show that deferred acceptance with minority reserves Pareto dominates deferred acceptance with majority quotas. Other contributions include \cite{abd2005}, \cite{kojima12}, \cite{westkamp10}, \cite{dogan16}, \cite{fratro:2017}, \cite{aytur2020}, and \cite{doganyildiz23}.

The analyses in \citet{ehayeyi14}, \citet{echyen12}, \citet{dogan17}, and \citet{dogan19} are closest to our application. They take type-specific bounds as given and study the induced choice rules. \cite{abdulkadiroglu2025market} axiomatize a many-to-one assignment mechanism that combines reserves and quotas with priorities, while \cite{imamura2020} characterizes reserve-and-quota rules as a compromise between meritocracy and diversity. We instead derive type-specific reserves and quotas from our notion of comparative diversity and show that the resulting selection lies on the diversity-merit frontier.

We also relate to more general models of distributional objectives. \cite{hakoyeyo2022} study choice rules generated by ordinally concave diversity functions. \cite{echenique2026distributionalpreferencesmarketdesign} consider a binary relation over sets of agents and characterize when the resulting distributional preferences can be implemented by greedy choice rules. and deferred acceptance. Both papers, however, take the underlying distributional criterion as primitive. We instead use majorization and an institution's target distribution to formulate the comparative diversity ordering itself.

Our model starts from  majorization \citep{hardy1952inequalities} and  Robin Hood transfers \citep{pigou1912wealth,dalton1920measurement}. The connection between diversity and majorization appears in the ecological-diversity literature. In particular, \cite{patil} define a diversity ordering based on transfers from more abundant to less abundant species. They show that this ordering is equivalent to majorization and use it to organize familiar diversity measures, including the Shannon and Simpson indices. Our paper builds on this perspective but shifts the benchmark from equal representation to an arbitrary target composition, applies the resulting framework to integer selection problems, and shows that the maximally diverse selections maximize every index in a broad class.

Subsequent work generalizes majorization in several directions. Weighted majorization assigns unequal weights to coordinates \citep{sherman1951}, while relative majorization and $d$-majorization compare distributions relative to a non-uniform benchmark \citep{ruch1980,joe1990,renes2016}. Our comparison is closer to the latter but measures departures from the benchmark in levels rather than ratios. As in the apportionment literature, our problem combines proportional targets with indivisibility \citep{balinskiyoung,pukelsheim}. Our level-based comparison reflects the structure of selection problems in which institutions allocate a fixed number of slots and adjust the composition of a group by replacing one selected agent with another. A one-slot excess or shortfall is therefore treated as the same discrete departure for every type, whereas a ratio-based comparison scales the same departure according to the type’s target share.

\cite{schoot2025} introduces $(a,b)$-majorization, which incorporates coordinate-specific scaling and shifting. Our comparison corresponds to the case with no scaling and with shifts determined by the target composition. For integer selection problems, however, \cite{schoot2025} restricts attention to integer-valued shifts, whereas the arbitrary target compositions we consider may generate non-integer shifts. This distinction is particularly relevant for small groups, where the institution must select indivisible agents, even though the target composition may be infeasible.

Finally, our selection problem also has a discrete-convex interpretation. For a fixed group size, feasible type distributions form an $M$-convex set, connecting our exchange and local-improvement results to standard results in discrete convex analysis \citep{murotashioura1999,murota2003}. However, our characterization of the frontier goes beyond what $M$-convexity alone delivers and provides the additional structure that makes implementation through type-specific reserves and quotas possible. We provide a more detailed discussion of this in \autoref{rem:dca}.

\section{Model}\label{sec:model}

\subsection{Preliminaries} \label{sec:preliminaries}
We use $\mathbb{Z}_+$ and $\mathbb{Z}_{++}$ to denote the nonnegative and positive integers, respectively. Similarly,  we denote the nonnegative and positive reals by $\mathbb R_+$ and $\mathbb R_{++}$. For any $n\in\mathbb Z_{++}$ and 
$x \in \mathbb R^n$, let  $\norm{x}\coloneqq \sum_{i=1}^n x_i$, and let $x_{[i]}$ denote the $i^{th}$ largest coordinate of $x$ so that $x_{[1]}\geq \cdots \geq x_{[n]}$. For any $i\in \{1,\ldots,n\}$, let $\chi_i\in \mathbb{Z}_+^n$  denote the unit vector where the $i$-th coordinate is one and all other coordinates are zero. The $(n-1)$-dimensional simplex is denoted by $\Delta^{n-1}\coloneqq \{x\in \mathbb R^n_+: \norm{x}=1\}$. For any two vectors $x,y\in \mathbb R^n$, we say that  $x$ \emph{majorizes} $y$, denoted as $x\succ y$, if 
\[
\sum_{i=1}^k x_{[i]} \geq \sum_{i=1}^k y_{[i]}
\]
for every $k\in \{1,\ldots,n\}$, with equality at $k=n$. We say that $x$ \emph{strictly majorizes} $y$, denoted as $x\succ_s y$, if $x \succ y$ but $y\mathrel{\slashed \succ} x$. 

\subsection{Setup}\label{sec:setup}
We consider an environment in which each individual has one of $n\in\mathbb{Z}_{++}$ possible types. A type may represent a single attribute or an intersection of several attributes. For example, a graduate-admissions committee may classify applicants according to research interests, methodological specialization, predoctoral work experience, or combinations of these characteristics. We index the $n$ possible types by $i \in \mathcal{N} \coloneqq \{1,\ldots,n\}$. 

We represent a finite population by its type distribution $x\in\mathbb{Z}+^n$, where $x_i$ is the number of type-$i$ agents and $\langle x\rangle$ is the population size. Throughout the paper, we refer to a vector $x\in\mathbb{Z}_+^n$ as a \emph{feasible} distribution, while we refer to any vector $x\in\mathbb{R}^n$ simply as a distribution. The extension to $\mathbb{R}^n$ is a mathematical device used to characterize binary relations; only feasible distributions in $\mathbb{Z}_+^n$ have a population interpretation. We return to this restricted domain in \Cref{sec:maximal}.

\section{Diversity as Majorization}\label{sec:majorization}
We propose a notion of comparative diversity for two equal-sized populations based on their types and a target distribution. As a starting point, suppose the target features equal representation by all types: let $u\in\Delta^{n-1}$ denote the uniform measure, with $u_i=1/n$ for every $i\in\mathcal{N}$. Given $\lambda\in\mathbb{Z}_{++}$ agents, the target distribution is $\lambda u$, where $\lambda u_i$ denotes the target number of type-$i$ agents. If $\lambda u\notin\mathbb{Z}_+^n$ (i.e., if $\lambda$ is not divisible by $n$), the target distribution is infeasible. Nevertheless, one can evaluate other distributions according to how well they approximate the target.

Now consider any $x \in \mathbb{R}^n$ such that $x_i - \norm{x}{u}_i > x_j - \norm{x}{u}_j$ for some types $i,j \in \mathcal{N}$.\footnote{Since $u_i=u_j$, the inequality is equivalent to $x_i>x_j$, but writing the comparison in terms of deviations from the target will be useful when we generalize the benchmark.} Relative to the target, type $i$ is more represented than type $j$. Suppose we shift a mass $\delta \geq 0$ from type $i$ to type $j$, while leaving all other types unchanged. This yields a new distribution $y = x + \delta \chi_j - \delta \chi_i$. Since $\norm{x} = \norm{y}$, the relevant target distribution for both $x$ and $y$ is $\norm{x}{u}$. If $\delta \in (0, x_i - x_j)$,  $y$ is closer than $x$ to the target distribution,\footnote{Formally, $||y-\norm{y}{u}||_2 < ||x-\norm{x}{u}||_2$.} suggesting that $y$ is more diverse than $x$ under the uniform-measure benchmark. By contrast, if $\delta = x_i - x_j$, then $y$ is merely a permutation of $x$ in which the coordinates of types $i$ and $j$ are swapped, suggesting that $x$ and $y$ are equally diverse. Such transformations of a distribution are called Robin Hood transfers. \cite{hardy1952inequalities} show that $y$ is obtained from $x$ through a sequence of Robin Hood transfers if and only if $x$ majorizes $y$. 

For many applications, however, the uniform benchmark is too restrictive. We instead introduce a generalized notion of diversity based on an arbitrary benchmark measure $r \in \Delta^{n-1}$. Given $\lambda \in\mathbb{Z}_{++}$ agents, the target distribution becomes $\lambda r$, with $\lambda r_i$ capturing the target number of type-$i$ agents. Once again, such a target distribution is infeasible if $\lambda r \notin \mathbb{Z}_+^n$, but it can nevertheless serve as the benchmark against which other distributions are compared.

In the selection problems we have in mind, the composition of a group is adjusted by replacing an agent of one type with an agent of another. Such a one-agent swap changes the deviations of the two affected types from their respective targets by one unit, regardless of their target shares. This motivates a notion of diversity that measures over- and under-representation in levels---the number of agents by which each type exceeds or falls short of its target---rather than in proportional terms that vary with the type's target share.

Given a measure $r\in\Delta^{n-1}$, distribution $x\in\mathbb{R}^n$, and type $i\in\mathcal{N}$, define the deviation of type $i$ from its target level by $d_i^r(x)\coloneqq x_i-\norm{x}r_i$. Thus, type $i$ is relatively more represented than type $j$ in $x$ whenever $d_i^r(x)\geq d_j^r(x)$.

\smallskip
\begin{definition}\label{def:rht}
Given a measure ${r}\in\Delta^{n-1}$ and two type distributions $x, y\in\mathbb{R}^n$, we say that $y$ is obtained from $x$ via an \emph{$r$-targeting Robin Hood transfer} if there exist types $i,j\in\mathcal{N}$ and a constant $\delta$ satisfying  $d_i^r(x)\geq d_j^r(x)$ and $0\leq \delta\leq d_i^r(x)-d_j^r(x)$ such that $y=x+\delta \chi_j-\delta\chi_i$.
\end{definition}
\smallskip

\autoref{def:rht} extends Robin Hood transfers from the uniform measure to an arbitrary measure $r$: the distribution $y$ is constructed by moving mass from a relatively more represented type $i$ to a relatively less represented type $j$, so that $y$ is weakly closer than $x$ to the target distribution $\norm{x}r$. 
\smallskip
\begin{definition}\label{def:diversity}
Given a measure $r \in \Delta^{n-1}$ and two type distributions $x, y \in \mathbb{R}^n$, we say that $y$ is \emph{more $r$-diverse} than $x$, denoted by $y \mathrel{\unrhd_r} x$, if there exists a finite sequence of distributions $\{z_k\}_{k=0}^K$ with $z_0 = x$, $z_K = y$, such that for each $k \in \{1, \ldots, K\}$, $z_k$ is obtained from $z_{k-1}$ via an $r$-targeting Robin Hood transfer. We say that $y$ is \emph{strictly more $r$-diverse} than $x$, denoted by $y \mathrel{\rhd_r} x$, if $y \mathrel{\unrhd_r} x$ and $x \mathrel{\slashed\unrhd_r} y$.
\end{definition}
\smallskip

We refer to the diversity preorder in \autoref{def:diversity} as \emph{$r$-targeting diversity}. 

\begin{proposition}\label{prop:majorization}
Fix a measure $r\in\Delta^{n-1}$, and let $T^r:\mathbb{R}^n\to\mathbb{R}^n$ be given by 
\[
T^r(x) \coloneqq x+\norm{x}(u-r).
\]
For any two type distributions $x, y\in\mathbb{R}^n$, $y\mathrel{\unrhd_r} x$ if and only if $T^{r}(x)\succ T^r(y)$.
\end{proposition}

\autoref{prop:majorization} formalizes the equivalence between $r$-targeting diversity and majorization. This equivalence implies that the binary relation $\unrhd_r$ is reflexive and transitive, i.e., a preorder, although it is neither antisymmetric nor complete. The transformation $T^r$ linearly translates distributions so that, for $\lambda$ agents, the target distribution under measure $r$ is mapped to the target distribution under the uniform measure, i.e., $T^r(\lambda r) = \lambda u$. Thus, for any distribution $x \in \mathbb{R}^n$ with $\norm{x} = \lambda$, we have $T^r(x) \succ T^r(\lambda r)$, making $\lambda r$ the most diverse distribution with $\lambda$ agents (if $\lambda r$ is feasible). Moreover, $T^r$ preserves the population size: for any distribution $x$, we have $\norm{x} = \norm{T^r(x)}$. Finally, when $r=u$, $T^r$ reduces to the identity function, and we recover the standard equivalence between majorization and Robin Hood transfers.

For the remainder of the paper, we fix a target measure $r \in \Delta^{n-1}$ and consider the restriction of the associated $r$-targeting diversity preorder to the subset of feasible distributions $\mathbb{Z}_+^n$. 

\section{Maximally Diverse Selections}\label{sec:maximal}
Consider a decision maker who selects $q\in\mathbb{Z}_{++}$ agents from a population with type distribution $x\in\mathbb{Z}_+^n$, with $q\leq \norm{x}$. Which type distributions of selected groups are maximally $r$-diverse?

We first define the \emph{budget set}
\[
\mathcal{B}(x,q)\coloneqq\left\{y\in\mathbb{Z}_+^n:y\leq x\text{ and }\norm{y}=q\right\}, 
\]
as the set of all feasible distributions that can be generated by picking $q$ agents from a population whose type distribution is $x$. Clearly, $\mathcal{B}(x,q)$ is non-empty.

We next identify the distributions in the budget set that are maximal with respect to $r$-targeting diversity. Because $\unrhd_r$ is generally incomplete, a maximal distribution is one that is undominated by any other distribution in the budget set; it need not be comparable with, let alone dominate, every other distribution in the budget set. Formally, the set of maximally $r$-diverse distributions is given by
\[
\mathcal{F}_r(x,q)\coloneqq\left\{y\in\mathcal{B}(x,q):\nexists z\in\mathcal{B}(x,q)\text{ such that }z\mathrel{\rhd_r}y\right\}.
\]
We refer to $\mathcal{F}_r(x,q)$ as the \emph{$r$-diversity frontier}. Because the budget set is finite and nonempty, $\mathcal{F}_r(x,q)$ is also nonempty. 

While our definition of the $r$-diversity frontier is order-theoretic, we also provide an equivalent characterization using a family of real-valued diversity indices. For any concave function $g:\mathbb{R}\to\mathbb{R}$, define
\[
\Phi_g^r(y)\coloneqq\sum_{i\in\mathcal{N}}g\big(T_i^r(y)\big).
\]
The mapping $y\mapsto\sum_{i\in\mathcal{N}}g(y_i)$ is a separable Schur-concave function whenever $g$ is concave. We compose this mapping with the transformation $T^r$ and interpret the resulting function $\Phi_g^r$ as a target-adjusted diversity index. Let $\mathcal{G}$ denote the class of concave functions. By the convex-function characterization of majorization \citep{marshallolkin}, $y \mathrel{\unrhd_r}z$ if and only if $\Phi_g^r(y)\geq\Phi_g^r(z)$ for every $g\in\mathcal{G}$, with strict inequality whenever $y\mathrel{\rhd_r} z$ and $g$ is strictly concave. We consider the corresponding class of diversity indices $\{\Phi_g^r:g\in\mathcal{G}\}$, where different choices of $g$ capture different ways of evaluating deviations from the target. For example, when $g(t)=-t^2$, maximizing $\Phi_g^r$ is equivalent to minimizing the squared Euclidean distance from the target.


\begin{proposition}\label{prop:quadratic-characterization}
For every feasible distribution $x\in\mathbb{Z}_+^n$ and desired group size $q\in\mathbb{Z}_{++}$ with $q\leq \norm{x}$, 
\[
\mathcal{F}_r(x,q)=\bigcap_{g\in\mathcal{G}}\argmax_{y\in  \mathcal{B}(x,q)}\Phi^r_g(y).
\]
\end{proposition}


\autoref{prop:quadratic-characterization} shows that the $r$-diversity frontier is exactly the set of distributions that simultaneously maximize every separable Schur-concave diversity index over the budget set. In fact, this conclusion extends beyond separable indices: for every Schur-concave $\Psi:\mathbb{R}^n\to \mathbb{R}$,
\[
\mathcal{F}_r(x,q)\subseteq\argmax_{y\in\mathcal{B}(x,q)}\Psi(T^r(y)),
\]
with equality if $\Psi$ is \emph{strictly} Schur-concave.\footnote{Both claims follow from \autoref{prop:quadratic-characterization}. If $y\in\mathcal{F}_r(x,q)$ and $z\in\mathcal{B}(x,q)$, then $\Phi^r_g(y)\geq\Phi^r_g(z)$ for every $g\in\mathcal{G}$, so $T^r(z)\succ T^r(y)$ by the convex-function characterization of majorization \citep[Proposition~4.B.1]{marshallolkin}, and hence $\Psi(T^r(y))\geq\Psi(T^r(z))$ for every Schur-concave $\Psi$. Conversely, if $y\notin\mathcal{F}_r(x,q)$, take any $z\in\mathcal{F}_r(x,q)$: then $\Phi^r_g(y)\leq\Phi^r_g(z)$ for every $g\in\mathcal{G}$, with strict inequality for some $g$, so $T^r(y)\succ_s T^r(z)$, and any strictly Schur-concave $\Psi$ satisfies $\Psi(T^r(z))>\Psi(T^r(y))$.} Thus, the frontier is not driven by any particular cardinal measure of diversity.




Based on \autoref{prop:quadratic-characterization}, we first establish that the frontier satisfies an exchange property: whenever a frontier distribution contains more than the minimum number of type-$i$ agents and fewer than the maximum number of type-$j$ agents across all frontier distributions, replacing one type-$i$ agent with a type-$j$ agent preserves maximality. Moreover, any two frontier distributions are equally $r$-diverse, which does not follow from maximality alone; in principle, they could be be incomparable.

\begin{proposition}\label{prop:frontier-indifference}
Fix a feasible distribution $x\in\mathbb{Z}_+^n$ and a desired group size $q\in\mathbb{Z}_{++}$ such that $q\leq\norm{x}$. 
\begin{enumerate}[label={$(\roman*)$}]
\item For any $y\in\mathcal{F}_r(x,q)$ and distinct $i,j\in\mathcal{N}$ such that $y_i>\min_{z\in\mathcal{F}_r(x,q)}z_i$ and $y_j<\max_{z\in\mathcal{F}_r(x,q)}z_j$, 
\[
y-\chi_i+\chi_j\in\mathcal{F}_r(x,q).
\]
Moreover, 
\[
y_i=\min_{z\in\mathcal{F}_r(x,q)}z_i+1,\quad y_j=\max_{z\in\mathcal{F}_r(x,q)}z_j-1,  \text{ and }\quad T^r_i(y)=T^r_j(y)+1.
\]
\item For any $y,z\in\mathcal{F}_r(x,q)$, we have $y\mathrel{\unrhd_r}z$ and $z\mathrel{\unrhd_r}y$.
\end{enumerate}
\end{proposition}

We next show that, despite the incompleteness of the $r$-targeting diversity relation, every maximally diverse distribution is strictly more $r$-diverse than every non-maximal distribution in the budget set. Moreover, every non-maximal distribution admits a local improvement: it can be made strictly more $r$-diverse by replacing one selected agent with an agent of another type.

\begin{proposition}\label{prop:frontier-ranking}
Fix a feasible distribution $x\in\mathbb{Z}_+^n$ and desired group size $q\in\mathbb{Z}_{++}$ such that $q\leq \norm{x}$. Let $y\in \mathcal{B}(x,q)\backslash\mathcal{F}_r(x,q)$. Then,
\begin{enumerate}[label={$(\roman*)$}]
\item There exist $i, j\in \mathcal N$ such that  $y+\chi_i-\chi_j\in \mathcal{B}(x,q)$ and $y+\chi_i-\chi_j\rhd_r y$. 
\item For every $z\in \mathcal{F}_r(x,q)$, $z \mathrel{\rhd_r} y$.
\end{enumerate} 
\end{proposition}

Taken together, \autoref{prop:frontier-indifference} and \autoref{prop:frontier-ranking} strengthen the maximality used to define the $r$-diversity frontier. Any two frontier distributions are equally $r$-diverse, while every frontier distribution is strictly more $r$-diverse than every distribution outside the frontier. The incompleteness of the $r$-targeting diversity relation can only play a role among non-maximal distributions. In particular, the frontier consists of \textit{optimal} distributions, rather than merely maximal ones. The following example illustrates the properties of the $r$-diversity frontier. 

\begin{example}\label{example:frontier}
Let $n=4$, let the population distribution be $x=(4,2,1,1)$, and let the desired group size be $q=6$. The budget set $\mathcal{B}(x,6)$ is then given by
\[
\left\{(4,2,0,0), (4,1,0,1), (3,2,0,1), (4,1,1,0), (3,2,1,0), (4,0,1,1), (3,1,1,1), (2,2,1,1)\right\}.
\]
Let $r=\left({5}/{12},1/4,1/4,{1}/{12}\right)$, so the target distribution is $qr=\left(5/2,3/2,3/2,1/2\right)$. The $r$-diversity frontier is given by
\[
\mathcal{F}_r(x,6)=\left\{
(3,2,1,0),
(3,1,1,1),
(2,2,1,1)
\right\}.
\]
The transformed frontier distributions are
\[
T^r(3,2,1,0)=(2,2,1,1), \quad T^r(3,1,1,1)=(2,1,1,2), \quad T^r(2,2,1,1)=(1,2,1,2).
\]
Because these transformed distributions are permutations of one another, the three frontier distributions are equally $r$-diverse.

Now consider $y=(4,1,1,0)$ and $z=(3,2,0,1)$. Their transformed distributions are
\[
T^r(y)=(3,1,1,1)\qquad\text{and}\qquad T^r(z)=(2,2,0,2),
\]
with neither transformed distribution majorizing the other. Thus, $y$ and $z$ are incomparable in terms of $r$-targeting diversity. Moreover, treating the normalized transformed distributions $T^r(y)/q$ and $T^r(z)/q$ as compositions,\footnote{Indices defined only on probability vectors are applied after a fixed common translation that makes every transformed vector in the budget set nonnegative, followed by normalization; both operations preserve majorization.} Shannon entropy ranks $y$ above $z$, the inverse Berger--Parker index ranks $z$ above $y$, and the Gini--Simpson index assigns them the same value. Finally, neither distribution is in the frontier---both $y$ and $z$ admit a strict local improvement:
\[
y+\chi_4-\chi_1=(3,1,1,1) \quad \text{ and } \quad z+\chi_3-\chi_4=(3,2,1,0),
\]
which shifts $y$ and $z$ into the frontier. 
\end{example}

\section{Reserves and Quotas}\label{sec:reserves-and-quotas}
We now apply our results to a selection problem in which the decision maker values diversity and merit. Institutions often pursue distributional objectives through reserves, setting aside type-specific slots or quotas, which cap the number selected from each type. We use the $r$-diversity frontier to endogenously derive type-specific reserves and quotas. We then characterize the optimality properties of the resulting reserve-and-quota policy.

Let $\mathcal{A}$ be a finite population of agents. Agents are ranked according to merit, which we refer to as \emph{priority}. Let $P$ denote the priority ranking, which is a strict linear order over $\mathcal{A}$. In addition to her priority, each agent has a type $\tau(a)$, where $\tau:\mathcal{A}\to\mathcal{N}$. For every subset $A\subseteq\mathcal{A}$, define its type distribution $\xi(A)\in\mathbb{Z}_+^n$ by
\[
\xi_i(A)\coloneqq\abs{\left\{a\in A:\tau(a)=i\right\}}
\]
for every $i\in\mathcal{N}$, with $\xi(\varnothing)=0$. Let $x\coloneqq\xi(\mathcal{A})$ denote the population-wide type distribution.

The decision maker selects a group of $q\in\mathbb{Z}_{++}$ agents, where $q\leq\abs{\mathcal{A}}$, and seeks to select the highest-priority agents according to $P$ while ensuring that the type distribution of the selected group is maximally $r$-diverse among the distributions in $\mathcal{B}(x,q)$. 

We first construct vectors of lower and upper bounds using the $r$-diversity frontier. To that end, define the lower-bound vector as the meet of the frontier
\[
L_r(x,q)\coloneqq \bigwedge\mathcal{F}_r(x,q)=\left(\min_{y\in\mathcal{F}_r(x,q)}y_1,\ldots,\min_{y\in\mathcal{F}_r(x,q)}y_n\right),
\]
and the upper-bound vector as the join of the frontier
\[
U_r(x,q)\coloneqq \bigvee\mathcal{F}_r(x,q)=\left(\max_{y\in\mathcal{F}_r(x,q)}y_1,\ldots,\max_{y\in\mathcal{F}_r(x,q)}y_n\right).
\]
Thus, every maximally $r$-diverse distribution that can be obtained by selecting $q$ agents from $\mathcal{A}$ contains at least $L_{r,i}(x,q)$ and at most $U_{r,i}(x,q)$ type-$i$ agents. By construction, $\norm{L_r(x,q)}\leq q\leq \norm{U_r(x,q)}$.  We interpret $L_r(x,q)$ and $U_r(x,q)$ as the reserve and quota vectors.

\begin{lemma}\label{lem:frontier-bounds}
For every feasible distribution $x\in\mathbb{Z}_+^n$ and desired group size $q\in\mathbb{Z}_{++}$ with $q\leq \norm{x}$, 
\[
\mathcal{F}_r(x,q)=\left\{y\in\mathcal{B}(x,q):L_r(x,q)\leq y\leq U_r(x,q)\right\}.
\]
\end{lemma}

From Part $(i)$ of \autoref{prop:frontier-indifference}, $U_{r,i}(x,q)-L_{r,i}(x,q)\in\{0,1\}$ for every type $i\in \mathcal{N}$. Hence, \autoref{lem:frontier-bounds} provides a \emph{unit-width box characterization} of the frontier. Before we formalize the reserve-and-quota policy, let us first make two remarks relating our characterizations of the frontier to the apportionment and discrete-convexity literature.

\begin{remark}\label{rem:apportionment}
When $x_i\geq q$ for every $i\in\mathcal{N}$, no type-specific availability constraint binds, and the $r$-diversity frontier coincides with the set of largest-remainder (Hamilton) apportionments of $q$ slots across types with entitlements $qr$. Each type receives either the floor or the ceiling of its entitlement, and ties in the fractional remainders generate multiple frontier distributions \citep{balinskiyoung,pukelsheim}. In contrast, when an availability constraint binds---as in \autoref{example:frontier}---the frontier may differ from the Hamilton apportionments. The resulting reserves and quotas then depend on the realized population $x$, rather than on $qr$ alone.
\end{remark}

\begin{remark}\label{rem:dca}
The budget set $\mathcal{B}(x,q)$ is $M$-convex, and $\Phi_g^r$ is $M$-concave on this set for each fixed $g\in\mathcal{G}$ \citep{murota2003}. Thus, for each separable index $\Phi_g^r$, $M$-concavity provides an index-specific local-improvement criterion, together with an exchange property for the index's maximizers. By contrast, \autoref{prop:quadratic-characterization} and its extension to arbitrary Schur-concave functions show that every distribution in the $r$-diversity frontier maximizes every Schur-concave index. Accordingly, \autoref{prop:frontier-indifference} and \autoref{prop:frontier-ranking} establish exchange and local-improvement properties that hold uniformly across all Schur-concave indices---including non-separable functions that need not be $M$-concave. Moreover, the unit-width box characterization of \autoref{lem:frontier-bounds} does not follow from $M$-concavity alone and is what makes implementation through type-specific reserves and quotas possible.
\end{remark}




We now describe the reserve-and-quota policy, which  selects agents as follows:
\begin{quoting}[leftmargin=.5cm, rightmargin=0cm]
\begin{description}[parsep=0pt, itemsep=1pt, leftmargin=0em]
\item[Process reserved slots:] For each $i\in\mathcal{N}$, fill the $L_{r,i}(x,q)$ reserved slots by selecting the highest-priority type-$i$ agents.
\item[Process open slots:] Fill the remaining $q-\norm{L_r(x,q)}$ open slots by selecting the highest-priority remaining agents, regardless of type, subject to selecting at most $U_{r,i}(x,q)$ type-$i$ agents in total for every $i\in\mathcal{N}$.
\end{description}
\end{quoting}

Let $A^{RQ}\subseteq\mathcal{A}$ denote the set of agents selected by the reserve-and-quota policy. As we show next, $A^{RQ}$ lies on the diversity-merit Pareto frontier. To that end, we first extend the priority ranking from individual agents to sets of agents. Given two sets $A,A'\subseteq\mathcal{A}$ with $\abs{A}=\abs{A'}=q$, label their respective agents in descending priority order: $a_1\mathrel{P}\cdots \mathrel{P} a_{q}$ and $a_1'\mathrel{P}\cdots \mathrel{P} a_{q}'$. We say that $A$ \emph{priority dominates} $A'$ if either $a_k=a_k'$ or $a_k \mathrel{P} a_k'$ for every $k=1,\ldots, q$. Intuitively, a set $A$ priority dominates $A'$ if, rank by rank, it contains agents with weakly higher priority.

\begin{proposition}\label{prop:optimality}
The agents selected by the reserve-and-quota policy, $A^{RQ}$, satisfy $\xi(A^{RQ})\in\mathcal{F}_r(x,q)$. Moreover, for any alternative selection $A\subseteq\mathcal{A}$ such that $\abs{A}=q$, at least one of the following holds:
\begin{enumerate}[label={$(\roman*)$}]
\item $\xi(A^{RQ})\mathrel{\rhd_r} \xi(A)$.
\item $A^{RQ}$ priority dominates $A$. 
\end{enumerate}
\end{proposition}

The reserve-and-quota policy reconciles the diversity and merit objectives by following a two-step procedure. The policy first restricts attention to selections with maximally $r$-diverse type distributions, then uses priority to select among them. \autoref{prop:optimality} establishes a global optimality property: no alternative selection improves upon the reserve-and-quota outcome according to either diversity or merit without sacrificing the other objective.

We conclude by returning to \autoref{example:frontier} to illustrate the reserve-and-quota policy. Let $\mathcal{A}=\{a_1,a_2,a_3,a_4,b_1,b_2,c_1,d_1\}$, where each $a_k$ is type $1$, each $b_k$ is type $2$, $c_1$ is type $3$, and $d_1$ is type 4. Thus, $\xi(\mathcal{A})=(4,2,1,1)=x$. Suppose the priority ranking is
\[
a_1\mathrel{P}b_1\mathrel{P}c_1\mathrel{P}a_2\mathrel{P}a_3\mathrel{P}a_4\mathrel{P}b_2\mathrel{P}d_1.
\]

Recall that $\mathcal{F}_r(x,6)=\left\{(3,2,1,0), (3,1,1,1),(2,2,1,1) \right\}$. Therefore, $L_r(x,6)=(2,1,1,0)$ and $U_r(x,6)=(3,2,1,1)$. The type-1 reserved slots are filled by $a_1$ and $a_2$; type-2 reserves are filled by $b_1$; and type-3 reserves are filled by $c_1$. There are two open slots, which are filled by $a_3$ and $b_2$. Agent $a_4$ is rejected despite having higher priority than $b_2$ because the type-$1$ quota is binding. Hence, $A^{RQ}=\{a_1,a_2,a_3,b_1,b_2,c_1\}$ with $\xi(A^{RQ})=(3,2,1,0)\in\mathcal{F}_r(x,6)$.

Every other six-agent selection either priority dominates $A^{RQ}$ but is strictly less $r$-diverse---for example, $\{a_1,a_2,a_3,a_4,b_1,c_1\}$---is equally $r$-diverse but priority dominated by $A^{RQ}$---for example, $\{a_1,a_2,a_3,b_1,c_1,d_1\}$---or is both priority dominated by and strictly less $r$-diverse than $A^{RQ}$---for example, $\{a_1,a_2,a_3,b_1,b_2,d_1\}$.

\newpage
\appendix
\section{Appendix}
\noindent \begin{proof}[Proof of \autoref{prop:majorization}]
Note that $T^r$ is invertible with $[T^r]^{-1}(z)\coloneqq z+\norm{z}(r-u)$. Consider two distributions $z, z'\in\mathbb{R}^n$ such that $z'$ is obtained from $z$ via an $r$-targeting Robin Hood transfer. Then there exist types $i,j\in\mathcal{N}$ with $d_i^r(z)\geq d_j^r(z)$, and a constant $0\leq \delta\leq d_i^r(z)- d_j^r(z)$ such that $z'=z+\delta \chi_j-\delta\chi_i$. 

Because $T^r$ preserves total mass,
\[
d_i^u(T^r(z))-d_j^u(T^r(z))=d_i^r(z)-d_j^r(z)\geq \delta
\]
and $T^r(z')=T^r(z)+\delta\chi_j-\delta\chi_i$. Thus, $z'$ is obtained from $z$ via an $r$-targeting Robin Hood transfer if and only if $T^r(z')$ is obtained from $T^r(z)$ via a standard ($u$-targeting) Robin Hood transfer, where the reverse direction follows from the same identities and the invertibility of $T^r$. Consequently, if $z'$ is obtained from $z$ via an $r$-targeting Robin Hood transfer, then $T^r(z)\succ T^r(z')$.

Now consider any $x,y\in\mathbb{R}^n$ such that $y\mathrel{\unrhd_r} x$. Then, there exists a finite sequence $\{z_k\}_{k=0}^K$ with $z_0 = x$, $z_K = y$, such that for each $k \in \{1, \ldots, K\}$, $T^r(z_{k-1})\succ T^r(z_{k})$. By transitivity, we conclude $T^r(x)\succ T^r(y)$.

Conversely, consider any $x,y\in\mathbb{R}^n$ such that $T^r(x)\succ T^r(y)$. Then there exists a finite sequence $\{z_k\}_{k=0}^K$ with $z_0=T^r(x)$ and $z_K=T^r(y)$ such that, for each $k$, $z_k$ is obtained from $z_{k-1}$ via a standard Robin Hood transfer \citep[Theorem B.2]{marshallolkin}. Then $[T^r]^{-1}(z_0) = x$, $[T^r]^{-1}(z_K) = y$, and for each $k$, $[T^r]^{-1}(z_k)$ is obtained from $[T^r]^{-1}(z_{k-1})$ via an $r$-targeting Robin Hood Transfer. Thus, $y\unrhd_r x$.
\end{proof}\medskip

\begin{lemma}\label{lem:quadratic-exchange}
Fix $x\in\mathbb{Z}_+^n$ and $q\in\mathbb{Z}_{++}$ with $q\leq\norm{x}$. For each $g\in\mathcal{G}$, define
\[
M^r_g(x,q)\coloneqq\argmax_{y\in\mathcal{B}(x,q)}\Phi^r_g(y).
\]
For any $y\in\mathcal{B}(x,q)$, $y\notin M^r_g(x,q)$ for some  $g\in\mathcal{G}$ if and only if there exist $i,j\in\mathcal{N}$ such that 
\[
y+\chi_j-\chi_i\in\mathcal{B}(x,q)\quad \text{ and } \quad y+\chi_j-\chi_i\mathrel{\rhd_r}y.
\]
\end{lemma}

\begin{proof}
Because $\mathcal{B}(x,q)$ is finite and nonempty, $M^r_g(x,q)$ is nonempty for any $g\in\mathcal{G}$. Fix $y\in\mathcal{B}(x,q)$.\medskip

\noindent\emph{``If'' direction:}
Suppose there exist $i,j\in\mathcal{N}$ satisfying the antecedent of the lemma. Let $y'\coloneqq  y+\chi_{j}-\chi_{i}$. Since $y' \mathrel{\rhd_r} y$, by \autoref{prop:majorization}, we have $T^r(y) \succ_s T^r(y')$. From \cite[Lemma C.1.a]{marshallolkin}, any strictly concave $g\in\mathcal{G}$ yields $\Phi^r_g(y')>\Phi^r_g(y)$. Thus, $y\notin M^r_g(x,q)$.  
\medskip

\noindent\emph{``Only if'' direction:} Suppose $y\notin M_g^r(x,q)$ for some $g\in\mathcal{G}$. Since $M_g^r(x,q)\neq\varnothing$, take any $z\in M_g^r(x,q)$. Define
\[
A\coloneqq\left\{i\in\mathcal{N}:y_i>z_i\right\}
\qquad\text{and}\qquad
B\coloneqq\left\{j\in\mathcal{N}:y_j<z_j\right\}.
\]
Because $y\neq z$ and $\norm{y}=\norm{z}=q$, both $A$ and $B$ are nonempty. For each $i\in A$, let $a_i\coloneqq y_i-z_i$ and, for each $j\in B$, let $b_j\coloneqq z_j-y_j$. Then
\[
\sum_{i\in A}a_i=\sum_{j\in B}b_j>0.
\]

Define the unit increment of $g$ by
\[
\nabla_g(t)\coloneqq g(t+1)-g(t).
\]
Because $g$ is concave, $\nabla_g$ is weakly decreasing. Moreover, because $\norm{y}=\norm{z}$, $T_i^r(z)=T_i^r(y)-a_i$ for every $i\in A$, and $T_j^r(z)=T_j^r(y)+b_j$ for every $j\in B$.   

Choose $i^*\in\argmax_{i\in A}T_i^r(y)$ and $j^*\in\argmin_{j\in B}T_j^r(y)$. For every $i\in A$ and $m\in\{1,\ldots,a_i\}$,
\begin{equation}\label{eq:eq1}
 T_i^r(y)-m\leq T_{i^*}^r(y)-1,   
\end{equation}
whereas, for every $j\in B$ and $m\in\{0,\ldots,b_j-1\}$,
\begin{equation}\label{eq:eq2}
    T_j^r(y)+m\geq T_{j^*}^r(y).
\end{equation}

Since $z\in M^r_g(x,q)$ and $y\notin M^r_g(x,q)$, we have
\begin{align*}
0&>\Phi^r_g(y)-\Phi^r_g(z)\\[6pt]
&=\sum_{\ell\in A\cup B}\left[g\big(T_\ell^r(y)\big)-g\big(T_\ell^r(z)\big)\right]\\[6pt]
&=\sum_{i\in A}\sum_{m=1}^{a_i}\nabla_g\big(T_i^r(y)-m\big)-\sum_{j\in B}\sum_{m=0}^{b_j-1}\nabla_g\big(T_j^r(y)+m\big)\\[6pt]
&\geq\nabla_g\big(T_{i^*}^r(y)-1\big)\sum_{i\in A}a_i-\nabla_g\big(T_{j^*}^r(y)\big)\sum_{j\in B} b_j\\[6pt]
&=\left(\nabla_g\big(T_{i^*}^r(y)-1\big)-\nabla_g\big(T_{j^*}^r(y)\big)\right)\sum_{i\in A} a_i,
\end{align*}
where the first equality follows because $T^r_i(y)=T^r_i(z)$ for all $i\notin A\cup B$, the second inequality follows from monotonicity  of $\nabla_g$ and \eqref{eq:eq1} and \eqref{eq:eq2}, and the last equality follows from the fact that $\sum_{i\in A} a_i=\sum_{j\in B}b_j>0$. Therefore, the above chain of inequalities implies that 
\[
\nabla_g\big(T_{i^*}^r(y)-1\big)<\nabla_g\big(T_{j^*}^r(y)\big)
\]
which is only possible if $T^r_{i^*}(y)-T^r_{j^*}(y)>1$.

Now define $y'\coloneqq y+\chi_{j^*}-\chi_{i^*}$. Clearly, $y'\in\mathbb{Z}_+^n$ and $\norm{y'}=\norm{y}=q$. Since $i^*\in A$ and $j^*\in B$, we have $y_{i^*}\geq 1$ and $y_{j^*}+1\leq x_{j^*}$. Hence, $y'\in\mathcal{B}(x,q)$. Finally, notice that $T^r(y')=T^r(y)+\chi_{j^*}-\chi_{i^*}$. Thus, $T^r(y')$ is obtained from $T^r(y)$ via a standard Robin Hood transfer with $\delta=1<T^r_{i^*}(y)-T^r_{j^*}(y)$. Hence, $T^r(y)\succ_sT^r(y')$. By \autoref{prop:majorization},  $y'\rhd_r y$.
\end{proof}\medskip

\begin{proof}[Proof of \autoref{prop:quadratic-characterization}]
Take any $y\in\cap_{g\in\mathcal{G}}M_g^r(x,q)$. If there exists a $z\in\mathcal{B}(x,q)$ such that $z\mathrel{\rhd_r}y$, then \autoref{prop:majorization} and strict Schur concavity imply $\Phi_g^r(z)>\Phi_g^r(y)$ for every strictly concave $g\in\mathcal{G}$, contradicting $y\in M_g^r(x,q)$. Hence, $\cap_{g\in\mathcal{G}}M_g^r(x,q)\subseteq\mathcal{F}_r(x,q)$.

Conversely, if $y\notin M_g^r(x,q)$ for some $g\in\mathcal{G}$, \autoref{lem:quadratic-exchange} yields a feasible distribution that is strictly more $r$-diverse than $y$. Thus, $y\notin\mathcal{F}_r(x,q)$. Taking the contrapositive yields $\mathcal{F}_r(x,q)\subseteq \cap_{g\in\mathcal{G}}M^r_g(x,q)$. 
\end{proof}\medskip



\begin{proof}[Proof of \autoref{prop:frontier-indifference}] Fix any $y\in\mathcal{F}_r(x,q)$.\medskip

\noindent \emph{Part $(i)$:} Suppose there exist distinct $i,j\in \mathcal{N}$ such that $y_i>\min_{w\in\mathcal{F}_r(x,q)}w_i$ and $y_j<\max_{w\in\mathcal{F}_r(x,q)}w_j$. Thus, there exist $z,z'\in\mathcal{F}_r(x,q)$ such that $\min_{w\in\mathcal{F}_r(x,q)}w_i=z_i$ and $\max_{w\in\mathcal{F}_r(x,q)}w_j=z'_j$.

Since $\norm{y}=\norm{z}$ and $y_i>z_i$, there exists $k\in\mathcal{N}$ such that $y_k<z_k$. Then, $y-\chi_i+\chi_k\in\mathcal{B}(x,q)$. If $T_i^r(y)-T_k^r(y)>1$, then $y-\chi_i+\chi_k\mathrel{\rhd_r}y$, contradicting $y\in\mathcal{F}_r(x,q)$. Thus,
\[
T_i^r(y)-T_k^r(y)\leq 1.
\]
An analogous argument establishes that
\[
T_k^r(z)-T_i^r(z)\leq 1.
\]
On the other hand, because $\norm{y}=\norm{z}$,
\[
T_i^r(y)-T_i^r(z)\geq1
\]
and
\[
T_k^r(z)-T_k^r(y)\geq1.
\]
Combining these four inequalities yields
\[
\Big[T_i^r(y)-T_i^r(z)\Big]+\Big[T_k^r(z)-T_k^r(y)\Big]=2.
\]
Because both terms in the square brackets are positive integers, each must equal one. Therefore,  
\begin{equation}\label{eq:oneequality}
    T^r_i(y)=T^r_i(z)+1=T^r_k(z)=T^r_k(y)+1, 
\end{equation}
and 
\[
T_i^r(y)-T_i^r(z)=1\Longleftrightarrow \quad y_i  =\underbrace{z_i+1}_{\min_{w\in\mathcal{F}_r(x,q)}w_i+1}.
\]
Similarly, since $\norm{y}=\norm{z'}$ and $y_j<z'_j$, there exists $k'\in\mathcal{N}$ such that $y_{k'}>z'_{k'}$,  
\begin{equation}\label{eq:twoequality}
T^r_{k'}(y)=T^r_{k'}(z')+1=T^r_j(z')=T^r_j(y)+1, 
\end{equation} 
and 
\[
T_j^r(z')-T_j^r(y)=1\Longleftrightarrow \quad y_j  =\underbrace{z'_j-1}_{\max_{w\in\mathcal{F}_r(x,q)}w_j-1}.
\]

Notice $y_{k'}>z_{k'}'$ and $y_k<z_k$ imply $y-\chi_{k'}+\chi_k\in\mathcal{B}(x,q)$. If $T_{k'}^r(y)-T_k^r(y)>1$, then $y-\chi_{k'}+\chi_k\rhd_r y$, contradicting $y\in\mathcal{F}_r(x,q)$. Thus,
\[
T_{k'}^r(y)-T_k^r(y)=T^r_j(y)+1-(T^r_i(y)-1)\leq 1,
\]
where the equality follows from \eqref{eq:oneequality} and \eqref{eq:twoequality}. Hence, $T^r_i(y)\geq T^r_j(y)+1$.

Similarly, $y_{i}>z_{i}$ and $y_j<z_j'$ imply $y-\chi_{i}+\chi_j\in\mathcal{B}(x,q)$. If $T_i^r(y)>T_j^r(y)+1$, then $y-\chi_i+\chi_j\mathrel{\rhd_r}y$, contradicting $y\in\mathcal{F}_r(x,q)$. Thus,
\[
T_i^r(y)=T_j^r(y)+ 1.
\]
Finally, $T^r(y-\chi_i+\chi_j)=T^r(y)-\chi_i+\chi_j$ is obtained by permuting coordinates $i$ and $j$ of $T^r(y)$. Hence, $T^r(y-\chi_i+\chi_j)$ and $T^r(y)$ majorize each other, implying $y-\chi_i+\chi_j\in\mathcal{F}_r(x,q)$, as desired.\medskip

\noindent\emph{Part $(ii)$:} Consider any $z\in\mathcal{F}_r(x,q)$. The result is trivial if $z=y$, so assume that $z\neq y$. By \autoref{prop:quadratic-characterization}, both $y$ and $z$ are maximizers of every separable Schur-concave function. By the convex-function characterization of majorization \citep[Proposition~4.B.1]{marshallolkin}, $T^r(y)\succ T^r(z)$ and $T^r(z)\succ T^r(y)$. By \autoref{prop:majorization}, $y$ and $z$ are equally $r$-diverse. 
\end{proof}\medskip

\begin{proof}[Proof of \autoref{prop:frontier-ranking}] Fix $y\in \mathcal{B}(x,q)\backslash\mathcal{F}_r(x,q)$ and $z\in \mathcal{F}_r(x,q)$.  \medskip

\noindent \emph{Part $(i)$:} By \autoref{prop:quadratic-characterization}, $y\notin M^r_g(x,q)$ for some $g\in\mathcal{G}$. The result then immediately follows from \autoref{lem:quadratic-exchange}.\medskip

\noindent \emph{Part $(ii)$:} 
 Since $z\in \mathcal{F}_r(x,q)$ and $y\notin \mathcal{F}_r(x,q)$, \autoref{prop:quadratic-characterization} implies that $\Phi^r_g(z)\geq \Phi^r_g(y)$ for all $g\in\mathcal{G}$, and strictly for at least one. Thus, by the convex-function characterization of majorization \citep[Proposition~4.B.1]{marshallolkin}, $T^r(y)\succ T^r(z)$. Moreover, $T^r(z)\mathrel{\slashed \succ} T^r(y)$: otherwise $T^r(y)$ and $T^r(z)$ would be permutations of one another, implying $\Phi^r_g(z)=\Phi^r_g(y)$ for every $g\in\mathcal{G}$. Hence $T^r(y)\succ_s T^r(z)$. By \autoref{prop:majorization}, we have $z\mathrel{\rhd_r} y$, as desired. 
\end{proof}\medskip

\begin{proof}[Proof of \autoref{lem:frontier-bounds}]
Suppose first that $y\in\mathcal{F}_r(x,q)$. By the definitions of the meet and join, $L_r(x,q)\leq y\leq U_r(x,q)$.

Conversely, let $y\in\mathcal{B}(x,q)$ with $L_r(x,q)\leq y\leq U_r(x,q)$. First, suppose $U_{r,i}(x,q)=L_{r,i}(x,q)$ for all $i\in\mathcal{N}$. Then $\mathcal{F}_r(x,q)$ is a singleton with $\mathcal{F}_r(x,q)=\{y\}$.

Next, suppose there exists a nonempty subset $I\subseteq\mathcal{N}$ such that $L_{r,i}(x,q)<U_{r,i}(x,q)$ if and only if $i\in I$. This implies that $\abs{\mathcal{F}_r(x,q)}>1$. Moreover, from Part $(i)$ of \autoref{prop:frontier-indifference}, $U_{r,i}(x,q)=L_{r,i}(x,q)+1$ for all $i\in I$.

Take any $z\in\mathcal{F}_r(x,q)$ such that $z\neq y$. Define
\[
A\coloneqq\left\{i\in \mathcal{N}:y_i>z_i\right\}\qquad\text{and}\qquad
B\coloneqq\left\{j\in \mathcal{N}:y_j<z_j\right\}.
\]
Since $z\neq y$ and $\norm{z}=\norm{y}$, $A$ and $B$ are nonempty. Moreover, $A\cup B\subseteq I$.

Because $L_r(x,q)\leq y \leq U_r(x,q)$ and $L_r(x,q)\leq z\leq U_r(x,q)$, and because the lower and upper bound vectors differ by one agent in any coordinate in $I$, we have $y_i=z_i+1$ for all $i\in A$ and $y_j+1=z_j$ for all $j\in B$. Using the linearity of $T^r$ and the fact that $\norm{y}=\norm{z}$, we have $T^r_i(y)=T^r_i(z)+1$ for all $i\in A$ and $T^r_j(y)+1=T^r_j(z)$ for all $j\in B$.

Furthermore, Part $(i)$ of \autoref{prop:frontier-indifference} implies that for $i\in A$ and $j\in B$, $T^r_j(z)= T^r_i(z)+1$. Consequently, for all $i\in A$ and $j\in B$,
\[
T^r_i(y)=T^r_i(z)+1=T^r_j(z)=T^r_j(y)+1.
\]

Finally, for all $\ell\notin A\cup B$, $T^r_\ell(y)=T^r_\ell(z)$, and 
\begin{align*}
0&=\norm{T^r(y)}-\norm{T^r(z)}=\sum_{i\in A} \underbrace{\Big(T^r_i(y)-T^r_i(z)\Big)}_{=1}+\sum_{j\in B} \underbrace{\Big(T^r_j(y)-T^r_j(z)\Big)}_{=-1},
\end{align*}
implying that $\abs{A}=\abs{B}$. 

Let $\widehat{\pi}:A\to B$ be any bijection, and define $\pi:\mathcal{N}\to\mathcal{N}$ by
\[
\pi(i)\coloneqq
\begin{cases}
i, & \text{if }i\notin A\cup B,\\
\widehat{\pi}(i), & \text{if }i\in A,\\
\widehat{\pi}^{-1}(i), & \text{if }i\in B.
\end{cases}
\]
For every $i\in\mathcal{N}$, we have $T_i^r(y)=T_{\pi(i)}^r(z)$. Thus, $T^r(y)$ and $T^r(z)$ are permutations of one another. By \autoref{prop:majorization}, $y\unrhd_r z$ and $z\unrhd_r y$. Since $z\in\mathcal{F}_r(x,q)$, and $y$ and $z$ are equally $r$-diverse, we conclude that $y\in\mathcal{F}_r(x,q)$.
\end{proof}\medskip

\begin{proof}[Proof of \autoref{prop:optimality}]
Because $U_r(x,q)\leq x$ and
\[
q-\norm{L_r(x,q)}\leq \norm{U_r(x,q)}-\norm{L_r(x,q)},
\]
there are enough remaining agents to fill all open slots without violating the quotas. Hence, $\abs{A^{RQ}}=q$ and $\xi(A^{RQ})\in\mathcal{B}(x,q)$. By construction, $L_r(x,q)\leq \xi(A^{RQ})\leq U_r(x,q)$. Hence, by \autoref{lem:frontier-bounds}, $\xi(A^{RQ})\in\mathcal{F}_r(x,q)$.

Next, consider any selection $A\subseteq\mathcal{A}$ such that $\abs{A}=q$ and $A\neq A^{RQ}$. Suppose first that $\xi(A)\notin\mathcal{F}_r(x,q)$. By Part $(ii)$ of \autoref{prop:frontier-ranking}, $\xi(A^{RQ})\mathrel{\rhd_r}\xi(A)$.


Next, suppose $\xi(A)\in\mathcal{F}_r(x,q)$. For ease of notation, let $L\coloneqq L_r(x,q)$ and $U\coloneqq U_r(x,q)$.

By Part $(i)$ of \autoref{prop:frontier-indifference}, $U_i-L_i\in\{0,1\}$ for all $i\in\mathcal{N}$. Define
\[
I\coloneqq\left\{i\in\mathcal{N}:U_i=L_i+1\right\}.
\]
For every $i\in\mathcal{N}$, let $S_i^R$ denote the set of the $L_i$ highest-priority type-$i$ agents, and let
\[
S^R\coloneqq\bigcup_{i\in\mathcal{N}}S_i^R
\]
denote the set of agents that the reserve-and-quota policy selects through the reserve slots. Thus, $S^R\subseteq A^{RQ}$ and $\abs{S^R}=\norm{L}$. 

For every $i\in I$, let $a_i^*$ denote the $(L_i+1)$-st highest-priority type-$i$ agent, and let 
\[
S^O\coloneqq\left\{a_i^*:i\in I\right\}
\]
denote the set of agents who are eligible for an open slot. Indeed, if the reserve-and-quota policy selects an additional type-$i$ agent, it would select $a_i^*$. Thus, the reserve-and-quota policy allocates the remaining $q-\norm{L}$ open slots to the $q-\norm{L}$ highest-priority agents in $S^O$.

Because $\xi(A)\in\mathcal{F}_r(x,q)$, \autoref{lem:frontier-bounds} implies $L\leq\xi(A)\leq U$. Define
\[
\mathcal{N}^1(A)\coloneqq\left\{i\in \mathcal{N}:\xi_i(A)=L_i+1\right\}.
\]
Note that $\mathcal{N}^1(A)\subseteq I$ and $\xi_i(A)=L_i$ for all $i\notin\mathcal{N}^1(A)$. Since $\norm{\xi(A)}=q$ and 
\[
\norm{\xi(A)}=\sum_{i\notin \mathcal{N}^1(A)}\xi_i(A)+\sum_{i\in \mathcal{N}^1(A)}\xi_i(A)=\norm{L}+\abs{\mathcal{N}^1(A)},
\]
we have $\abs{\mathcal{N}^1(A)}=q-\norm{L}$.  

Now define a new selection
\[
\widehat{A}\coloneqq S^R\cup\left\{a\in S^O: \tau(a)\in \mathcal{N}^1(A)\right\}.
\]
By construction, $\xi_i(\widehat{A})=L_i$ for all $i\notin \mathcal{N}^1(A)$ and $\xi_i(\widehat{A})=L_i+1$ for all $i\in \mathcal{N}^1(A)$. Thus, $\xi(A)=\xi(\widehat{A})$. Moreover, for each type $i$, $\widehat{A}$ contains the $\xi_i(A)$ highest-priority type-$i$ agents. Therefore, $\widehat{A}$ priority dominates $A$.

The sets $\widehat{A}$ and $A^{RQ}$ contain the same agents in $S^R$ and differ only in the allocation of the remaining $q-\norm{L}$ open slots. By construction, $A^{RQ}$ assigns those slots to the $q-\norm{L}$ highest-priority agents in $S^O$, whereas $\widehat{A}$ assigns them to the subset of agents in $S^O$ whose types are in $\mathcal{N}^1(A)$. Hence, $A^{RQ}$ priority dominates $\widehat{A}$. By transitivity, $A^{RQ}$ priority dominates $A$.
\end{proof}

\newpage
\singlespacing
\bibliographystyle{plainnat}
\nocite{}
\bibliography{matching}
\end{document}